\def\year{2018}\relax
\documentclass[letterpaper]{article}
\usepackage{aaai18}
\usepackage{times}
\usepackage{verbatim}
\usepackage{graphicx}
\usepackage{subfigure}
\usepackage{multirow}
\usepackage{algorithmic}
\usepackage{longtable}
\usepackage{array}
\usepackage{stfloats}
\usepackage{balance}
\usepackage{multicol}
\usepackage{color}
\usepackage{epstopdf}
\usepackage{bm}
\usepackage{amsmath}
\usepackage{amssymb}
\usepackage{amsthm}
\usepackage{booktabs} 
\usepackage{epsfig}
\usepackage{enumitem}
\usepackage{cleveref}
\usepackage{xspace}
\usepackage{url}
\usepackage{float}
\usepackage{booktabs}
\usepackage{makecell}
\usepackage{float}
\usepackage{helvet}
\usepackage{couriers}
\usepackage{url}
\newcommand{\hide}[1]{} 
\newcommand{\vpara}[1]{\noindent\textbf{#1 }}
\newcommand{\para}[1]{\noindent\textbf{#1 }}
\newcommand{\beq}[1]{\begin{equation}#1\end{equation}}

\DeclareMathOperator*{\argmin}{arg\,min}
\DeclareMathOperator*{\diag}{diag}
\DeclareMathOperator*{\trace}{trace}

\theoremstyle{definition}
\newtheorem{definition}{Definition}

\newtheorem*{remark}{Remark}
\newtheorem{theorem}{Theorem}

\title{Representation Learning for Scale-free Networks}

\author{
Rui Feng$^*$, Yang Yang$^\dag$\thanks{Equal contribution. Ordering determined by dice rolling.}, Wenjie Hu, Fei Wu, and Yueting Zhuang\\
College of Computer Science and Technology, Zhejiang University, China\\
$^\dag$Corresponding author: yangya@zju.edu.cn\\
}

\begin{document}
\maketitle

\begin{abstract}
Network embedding aims to learn the low-dimensional representations of vertexes in a network, while structure and inherent properties of the network is preserved. 
Existing network embedding works primarily focus on preserving the microscopic structure, such as the first- and second-order proximity of vertexes, while the macroscopic scale-free property is largely ignored.
Scale-free property depicts the fact that vertex degrees follow a heavy-tailed distribution (i.e., only a few vertexes have high degrees) and is a critical property of real-world networks, such as social networks.
In this paper, we study the problem of learning representations for scale-free networks. 
We first theoretically analyze the difficulty of embedding and reconstructing a scale-free network in the Euclidean space, by converting our problem to the sphere packing problem. 
Then, we propose the ``degree penalty'' principle for designing scale-free property preserving network embedding algorithm: punishing the proximity between high-degree vertexes. 
We introduce two implementations of our principle by utilizing the spectral techniques and a skip-gram model respectively. 
Extensive experiments on six datasets show that our algorithms are able to not only reconstruct heavy-tailed distributed degree distribution, but also outperform state-of-the-art embedding models in various network mining tasks, such as vertex classification and link prediction. 
\end{abstract}

\section{Introduction}	
\label{sec:intro}
Network analysis has attracted considerable research efforts in many areas of artificial intelligence, as networks are able to encode rich and complex data, such as human relationships and interactions. 
One major challenge of network analysis is how to represent network properly so that the network structure can be preserved.
The most straightforward way is to represent the network by its adjacency matrix. 
However, it suffers from the data sparsity. 
Other traditional network representation relies on handcrafted network feature design (e.g., clustering coefficient), which is inflexible, non-scalable, and requires hard human labor. 

In recent years, network representation learning, also known as network embedding, has been proposed and aroused considerable research interest. 
It aims to automatically project a given network into a low-dimensional latent space, and represent each vertex by a vector in that space. 
For example, a number of recent works apply advances in natural language processing (NLP), most notably models known as word2vec~\cite{mikolov2013efficient}, to network embedding and propose word2vec-based algorithms, such as DeepWalk~\cite{perozzi2014deepwalk} and node2vec~\cite{grover2016node2vec}. 
Besides, researchers also consider network embedding as part of dimensionality reduction techniques.
For instance, Laplacian Eigenmap (LE)~\cite{belkin2003laplacian} aims to learn the low-dimensional representation to expand the manifold  where the data lie. 

Essentially, these methods mainly focus on preserving microscopic structure of network, like pairwise relationship between vertexes.
Nevertheless, scale-free property, one of the most fundamental macroscopic properties of networks, is largely ignored. 

\begin{figure}[t]
\centering
\subfigure[Original]{
\label{fig:deg_citation_original}
\includegraphics{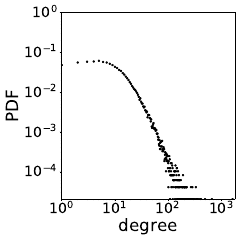}
}
\subfigure[LE]{
\label{fig:deg_citation_recovered_unmod}
\includegraphics{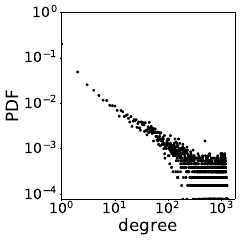}
}
\subfigure[DP-Walker]{
\label{fig:deg_citation_recovered}
\includegraphics{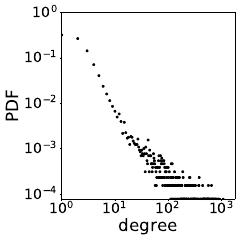}
}
\caption{Scale-free property of real-world networks. (a) is the degree distribution of an academic network. (b) and (c) are respectively the degree distribution of 
of the network reconstructed based on vertex representations learned by Laplacian Eigenmap (LE) and our proposed method (DP-Walker).
\normalsize }
\label{fig:degree}
\end{figure}

Scale-free property depicts that the vertex degrees follow a power-law distribution, which is a common knowledge for many real-world networks. 
We take an academic network 
as the example, where each edge indicates if a vertex (researcher) has cited at least one publication of another vertex (researcher). 
Figure~\ref{fig:deg_citation_original} demonstrates the degree distribution of this network. 
The linearity on log-log scale suggests a power-law distribution: the probability decreases when the vertex degree grows, with a long tail tending to zero~\cite{faloutsos1999power,adamic1999nature}. 
In other words, there are only a few vertexes of high degrees. 
The majority of vertexes connected to a high-degree vertex is, however, of low degree, and not likely connected to each other. 

Moreover, compared with the microscopic structure, the macroscopic scale-free property imposes a higher level constraint on the vertex representations: 
only a few vertexes can be close to many others in the learned latent space. 
Incorporating scale-free property in network embedding can reflect and preserve the sparsity of real-world networks and, in turn, provide effective information to make the vertex representations more discriminative.

In this paper, we study the problem of learning scale-free property preserving network embedding. 
As the representation of a network, vertex embedding vectors are expected to well reconstruct the network. 
Most existing algorithms learn network embedding in the Euclidean space.
However, 
we find that most traditional network embedding algorithms will overestimate the number of higher degrees.
Figure~\ref{fig:deg_citation_recovered_unmod} gives an example, where the representation is learned by Laplacian Eigenmap. 
We analyze and try to understand this theoretically, and study the feasibility of recovering power-law distributed vertex degree in the Euclidean space, by converting our problem to the Sphere Packing problem. 
Through our analysis we find that theoretically, moderately increasing the dimension of embedding vectors can help to preserve the scale-free property. 
See details in Section~\ref{sec:approach}. 


Inspired by our theoretical analysis, we then propose the \textit{degree penalty} principle for designing scale-free property preserving network embedding algorithms in Section~\ref{sec:model}: 
punishing the proximity between vertexes with high degrees. 
We further introduce two implementations of our approach based on the spectral techniques~\cite{belkin2003laplacian} and the skip-gram model~\cite{mikolov2013efficient}. 
As Figure~\ref{fig:deg_citation_recovered} suggests, our approach can better preserve the scale-free property of networks.
In particular, the Kolmogorov-Smirnov (aka. K-S statistic) distance between the obtained degree distribution and its theoretical power-law distribution is $0.09$, while the value for the degree distribution obtained by LE is $0.2$ (the smaller the better).   

To verify the effectiveness of our approach, we conduct experiments on both synthetic data and five real-world datasets in Section~\ref{sec:exp}. 
The experimental results show that our approach is able to not only preserve the scale-free property of networks, 
but also outperform several state-of-the-art embedding algorithms in different network analysis tasks. 

We summarize our contribution of this paper as follows:

\begin{itemize}
\item We analyze the difficulty and feasibility of reconstructing a scale-free network based on learned vertex representations in the Euclidean space, by converting our problem to the Sphere Packing problem. 
\item We propose the \textit{degree penalty} principle and two implementations 
to preserve the scale-free property of networks and 
improve the effectiveness of vertex representations. 
\item We validate our proposed principle by conducting extensive experiments and find that our approach achieves a significant improvement on six datasets and three tasks compared to several state-of-the-art baselines. 
\end{itemize}

\hide{
The rest of the paper is organized as follows. 
We give theoretical analysis in Section~\ref{sec:approach}, based on which we propose and introduce our method in Section~\ref{sec:model}. 
In Section~\ref{sec:exp}, we report the experimental results and then review the related work in Section~\ref{sec:related}. 
At last, Section~\ref{sec:conclusion} concludes the paper.
}

\section{Theoretical Analysis}
\label{sec:approach} 
In this section, we try to study why most network embedding algorithms will overestimate higher degrees theoretically, 
and analyze if there exists a solution of scale-free property preserving network embedding in the Euclidean space. 
This section also provides intuitions of our approach in Section~\ref{sec:model}. 

\subsection{Preliminaries}
\label{sec:prelim}

\vpara{Notations.}
We consider an undirected network $G=(V, E)$, where $V=\{v_1, \cdots, v_n\}$ is the vertex set containing $n$ vertexes and $E$ is the edge set. Each $e_{ij} \in E$ indicates an undirected edge between $v_i$ and $v_j$. 
We define the adjacency matrix of $G$ as $\mathbf{A}=[A_{ij}] \in \mathcal{R}^{n \times n}$, where $A_{ij} = 1$ if $e_{ij}\in E$ and $A_{ij} = 0$ otherwise.
Let $\mathbf{D}$ be a diagonal matrix where $D_{ii}=\sum_{j}A_{ij}$ is the degree of $v_i$. 

\vpara{Network embedding.} 
In this paper, we focus on the representation learning for undirected networks. 
Given an undirected graph $G=(V, E)$, the problem of graph embedding aims to represent each vertex $v_i \in V$ into a low-dimensional space $R^k$, i.e., learning a function $\mathbf{f}:V \mapsto\mathbf{U}^{n\times k}$, where $\mathbf{U}$ is the embedding matrix, $k \ll n$ and network structures can be preserved in $\mathbf{U}$. 

\hide{
There are various ways to quantify to what extent network structures have been preserved in $\mathbf{U}$. 
For instance, some methods like DeepWalk~\cite{perozzi2014deepwalk} and node2vec~\cite{grover2016node2vec} take into account paths in $G$ by assuming vertexes appearing in the same path tend to be more similar. 
Besides, Laplacian Eigenmaps  (LE)~\cite{belkin2003laplacian} aims to learn the low-dimensional representation to expand the manifold where vertexes lie. 
Few work study a more fundamental perspective of network structure: degree distribution, which we will introduce later. 
}

\vpara{Network reconstruction.}
As the representation of a network, the learned embedding matrix is expected to well reconstruct the network. 
In particular, one can reconstruct the network edges based on distances between vertexes in the latent space $\mathcal{R}^k$. 
For example, the probability of an edge existing between $v_i$ and $v_j$ can be defined as 

\beq{
\label{eq:reconstruct}
	p_{i,j} = \frac{1}{1+e^{\Vert \mathbf{u_i} - \mathbf{u_j} \Vert}} 
}

\noindent where the Euclidean distance between embedding vectors 
$\mathbf{u_i}$ and $\mathbf{u_j}$ of the vertex $v_i$ and $v_j$, in respective, 
is considered.
In practice, a threshold $\varepsilon \in [0, 1]$ is chosen and an edge $e_{ij}$ will be created if $p_{i, j} \geq \varepsilon$. 
We call the above method as $\varepsilon$-NN, which is geometrically informative and a common method used in many existing work \cite{shaw2009spe,belkin2003laplacian,alanis2016labne}.
 \begin{figure*}[t]
 \centering
   \includegraphics{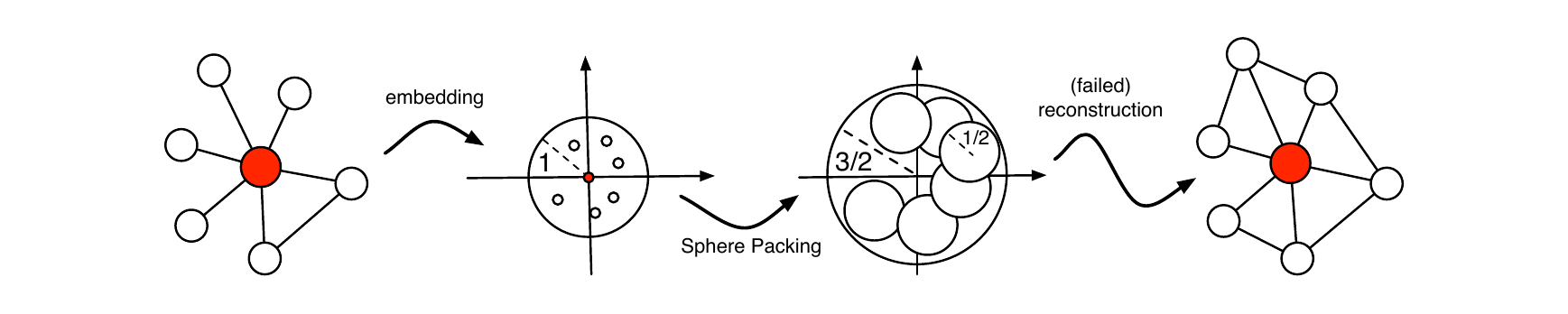}
   \caption{An illustration of an ego-network centered by a hub vertex, a potential embedding solution, which is equivalent to a sphere packing solution, and leads to a failed reconstruction, where higher degrees are overestimated. }
   \label{fig:example}
 \end{figure*}

\subsection{Reconstructing Scale-free Networks}
Given a network, we call it as a scale-free network, when its vertex degrees follow a power-law distribution.
In other words, 
there are only a few vertexes of high degrees. 
The majority of vertexes connected to a high-degree vertex is, however, of low degree, and not likely connected to each other. 
Formally, the probability density function of vertex degree $D_{ii}$ has the following form:

\beq{
\label{eq:powerlaw}
	 p_{D_{ii}}(d) = C d^{-\alpha} , \alpha > 1, d>d_{\min{}} > 0  
}

\noindent where $\alpha$ is the exponent parameter and $C$ is the normalization term. 
In practice, the above power-law form only applies for vertexes with degrees greater than a certain minimum value $d_{\min}$~\cite{aaron2009}. 

However, it is difficult to reconstruct a scale-free network in the Euclidean space by $\varepsilon$-NN. 
As we see in Figure~\ref{fig:deg_citation_recovered_unmod}, the higher degrees will be overestimated. 
We aim to explain this theoretically.

\vpara{Intuition.}
While reconstructing the network by $\varepsilon$-NN, we select a certain $\varepsilon$, and for a vertex $v_i$ with embedding vector $\mathbf{u_i}$, we regard all points that fall in the closed ball of radius $\varepsilon$ centered at $\mathbf{u}_i$, denoted by $B(\mathbf{u_i}, \varepsilon)$, as ones having edges with $v_i$.
When $v_i$ is a high-degree vertex, there will be many points in $B(\mathbf{u_i}, \varepsilon)$. 
We expect these points are far away from each other, keeping more vertexes with low degree and thus in turn keep the power-law distribution. 
However, intuitively, as these points are placed in the same closed ball $B(\mathbf{u_i}, \varepsilon)$, it will be more likely that their distances from each other are less than $\varepsilon$. 
As a result, there will be many edges created among those points, which violates the assumption of the scale-free property.

Following the above idea, we introduce a theorem below, which is discussed in $\mathcal{R}^k$, and
without loss of generality, we set $\varepsilon$ as $1$.

\begin{theorem}[Sphere Packing]
\label{theorem:packing}
There are $m$ points in $B(0, 1)$ whose distances from each other are larger than or equal to $1$, if and only if, there exists $m$ disjoint spheres of radius $1/2$ in $B(0,3/2)$.
\label{sp_thm1}
\end{theorem}

\begin{proof}
The center of any sphere of radius $1/2$ in $B(0, 3/2)$ falls in $B(0, 1)$ and the distance between any two centers of two different spheres of radius $1/2$ is larger than or equal to 1.  
\end{proof}

\begin{remark}
Theorem~\ref{theorem:packing} converts our problem of reconstructing a scale-free network to the \textit{Sphere Packing Problem}, which seeks to find the optimal packing of spheres in high dimensional spaces~\cite{cohn2003new,vance2011sphere,venkatesh2012sphere}. 
Figure~\ref{fig:example} gives an example, where a network centered by a high-degree vertex is embedded into a two-dimensional space. 
The embedding result corresponds to an equivalent Sphere Packing solution, which fails to place enough disjoint spheres and leads to a failed network reconstruction (many nonexistent edges are created). 
Formally, we define the packing density as follows: 
\end{remark}

\begin{definition}[Packing Density]
The packing density is the fraction of the space filled by the spheres making up the packing. 
For convenience, we define the optimal sphere packing density in $\mathcal{R}^k$ as $\Delta_k$, which means no packing of spheres in $\mathcal{R}^k$ achieves a packing density larger than $\Delta_k$.
\end{definition}

As Theorem~\ref{theorem:packing} suggests, we aim to find a packing solution with large density so that more points in a closed sphere can keep their distance larger than 1. 
However, in general cases, finding the optimal packing density 
remains an open problem. 
Still, 
we are able to derive the upper and lower bounds of $\Delta_k$ for sufficiently large $k$.

\begin{theorem}[Upper and lower bounds for $\Delta_k$]
$$ -1\leq \lim_{k\rightarrow\infty} \frac{1}{k}\log_2\Delta_k\leq-0.599 $$
\normalsize Specifically, we have 
\begin{equation}
\Delta_k \geq 2^{-k}, k \geq 1
\label{lower_bound}
\end{equation}
And the following inequality holds for sufficiently large $k$:
\begin{equation}
\Delta_k \leq 2^{-0.599k} \label{upper_bound}
\end{equation}
\end{theorem}

\noindent Eq.~\ref{upper_bound} is one of the best upper bound when $k \geq 115$~\cite{cohn2014sphere}.
The proof of the above theorem can be found in the work of~\citeauthor{kabatiansky1978bounds}.

\begin{theorem}

Suppose $x$ is in $\mathcal{R}^k$, $\varepsilon > 0$, and there can be at most $M_k$ points in $B(x, \varepsilon)$ whose distances from each other are larger than $\varepsilon$, then

\begin{eqnarray}
 \left\lfloor\left(\frac{3}{2}\right)^k\right\rfloor \leq M_k \leq \left\lfloor 3^k 2^{-0.599k}\right\rfloor  \label{bound_M}  \end{eqnarray}
where $\left\lfloor\cdot\right\rfloor$ means taking the integer part. The upper bound holds for sufficiently large 
\end{theorem}

\begin{proof}
By Theorem~\ref{sp_thm1} we only need to estimate the number of disjoint spheres of radius $\frac{1}{2}\varepsilon$ that can be fitted in $B(x, \frac{3}{2}\varepsilon)$. 
The volume of a $k$-dimensional ball of radius $R$ is given by $\frac{\pi^{k/2}}{\Gamma(k/2+1)}R^k$. 
Plugging in the radius, the volumes of $B(x, \frac{3}{2}\varepsilon)$ and a sphere of radius $\frac{1}{2}\varepsilon$ are respectively 
$V_1 = \frac{\pi^{k/2}}{\Gamma(k/2+1)}(\frac{3}{2}\varepsilon)^k$ and $V_2 = \frac{\pi^{k/2}}{\Gamma(k/2+1)}(\frac{1}{2}\varepsilon)^k$. 
Since the optimal packing density is given by $\Delta_k$, we have 
\begin{eqnarray*} 
M_k  =  \left\lfloor \Delta_k \frac{V_1}{V_2} \right\rfloor 
     =  \left\lfloor 3^k \Delta_k\right\rfloor
\end{eqnarray*}
 
\noindent Combing with Eq.~\ref{lower_bound} and Eq.~\ref{upper_bound}, we obtain the inequality as desired. 
 \end{proof}

\vpara{Discussion.}
The lower bound of $M_k$ in Eq.~\ref{bound_M} suggests the feasibility to reconstruct scale-free network by $\varepsilon$-NN in the Euclidean space, when $k$, the dimension of embedding vector, is sufficiently large. 
For instance, when $k=100$, we have $M_k > 4.06\times 10^{17}$, which is enough to keep scale-free property holds for most real-world networks. 

\section{Our Approach}
\label{sec:model}

\vpara{General idea.} Inspired by our theoretical analysis, we propose a principle of \textit{degree penalty} for designing scale-free property preserving embedding algorithms: \textit{while preserving first- and second-order proximity, the proximity between vertexes that have high degrees shall be punished}. 
We give the general idea behind this principle below. 

Scale-free property is featured by the ubiquitous existence of ``big hubs'' which attract most edges in the network. 
Most of existing network embedding algorithms, explicitly or implicitly, attempt to keep these big hubs close to their neighbors by preserving first-order and/or second-order proximity~\cite{belkin2003laplacian,tang2015line,perozzi2014deepwalk}.
However, connecting to big hubs does not imply proximity as strong as connecting to vertexes with mediocre degrees. 
Taking social network as an example, where a famous celebrity may receive a lot of followers.
However, the celebrity may not be familiar or similar to her followers. 
Besides, two followers of the same celebrity may not know each other at all and can be totally dissimilar. 
As a comparison, a mediocre user is more likely to known and to be similar to her followers. 

From another perspective, a high-degree vertex $v_i$ is more likely to hurt the scale-free property, as placing more disjoint spheres in a closed ball is more difficult (Section~\ref{sec:approach}). 

The \textit{degree penalty} principle can be implemented by various methods.
In this paper, we introduce two proposed models based on our principle, implemented by spectral techniques and skip-gram models respectively. 

\subsection{Model I: DP-Spectral}
Our first model, Degree Penalty based Spectral Embedding (DP-Spectral), mainly utilizes graph spectral techniques.  
Given a network $G=(V, E)$ and its adjacency matrix $\mathbf{A}$, we define a matrix $\mathbf{C}$ to indicate the common neighbors of any two vertexes in $G$: 

\beq{
\label{eq:common_neighbors}
\mathbf{C} \triangleq \mathbf{A}^T\mathbf{A} - \diag{(\mathbf{A}^T\mathbf{A})}
}

\noindent where $C_{ij}$ is the number of common neighbors of $v_i$ and $v_j$. 
$\mathbf{C}$ can be also regarded as a measurement of second-order proximity, and can be easily extended to further consider first-order proximity by 

\beq{
\label{eq:matrix_c}
\mathbf{C}' \triangleq \mathbf{C} + \mathbf{A}
}

\noindent As we aim to model the influence of vertex degrees in our model, we further extend $\mathbf{C}'$ to consider degree penalty as

\beq{
\label{eq:matrix_w}
	\mathbf{W} \triangleq (\mathbf{D}^{-\beta})^T \mathbf{C}' \mathbf{D}^{-\beta} 
}

\noindent where $\beta \in \mathcal{R}$ is the model parameter used to indicate the strength of degree penalty, and $\mathbf{D}$ is a diagonal matrix where $D_{ii}$ is the degree of $v_i$.
Thus $\mathbf{W}$ is proportional to $\mathbf{C}'$ and is inversely proportional to vertex degrees. 

\vpara{Objective.}
Our goal is to learn vertex representations $\mathbf{U}$, where $\mathbf{u_i} \in \mathcal{R}^k$ is the $i$th row of $\mathbf{U}$ and represents the embedding vector of vertex $v_i$, and minimize

\beq{
\label{eq:objective_dps}
	\sum_{i,j} \Vert\mathbf{u}_i - \mathbf{u}_j\Vert^2 W_{ij} 
}

\noindent under the constraint 

\beq{
\label{eq:constraint_dps}
\mathbf{U}^T\mathbf{DU}=\mathbf{I}
}

\noindent Eq.~\ref{eq:constraint_dps} is provided such that the embedding vectors will not collapse onto a subspace of dimension less than $k$~\cite{belkin2003laplacian}. 

\vpara{Optimization.}
In order to minimize Eq.~\ref{eq:objective_dps}, we utilize graph Laplacian, which is an analogy to the Laplacian operator in multivariate calculus. 
Formally, we define graph Laplacian, $\mathbf{L}$, as follows:

\beq{
\label{eq:laplacian_def}
\mathbf{L\triangleq D - W}
}

\hide{
The power of graph laplacian lies in the fact that for for any real-valued vector $\mathbf{x}\in\mathcal{R}^n$, we have 

\begin{equation}
	\mathbf{x}^TL\mathbf{x} = \sum_{\{i,j\}\in E} (x_i-x_j)^2 W_{i,j}\label{laplacian:laplacian}
\end{equation}

Graph Laplacian provides a sense of "smoothness" of a function defined on vertices of the graph, the values of which are the components of $\mathbf{x}$, with respect to the connectivity of $G$.~\cite{kolaczyk2009statmodelsofnetwork} With Eq. (\ref{laplacian:laplacian}), the optimal of the objective (\ref{laplacian:objective}) under the constraint (\ref{laplacian:constraint}) is given by the $k+1$ eigenvectors of the lowest eigenvalues, discarding the one corresponding to $0$. 
}
\noindent Observe that 

\beq{
\label{eq:lap_and_obj}
\sum_{i,j} W_{ij}\Vert \mathbf{u}_i - \mathbf{u}_j\Vert^2 = \trace{(\mathbf{U^TLU})} 
}

\noindent The desired $\mathbf{U}$ minimizing the objective function (\ref{eq:objective_dps}) is obtained by putting 

\begin{equation}
	\mathbf{U} = [\mathbf{t}_1, \ldots, \mathbf{t}_k]
\end{equation}

\noindent where $\mathbf{t}_i$ is an eigenvector of $\mathbf{L}$.
In practice, we use normalized form of $\mathbf{W}$ and $\mathbf{L}$, (i.e., $\mathbf{D}^{-1/2}\mathbf{WD}^{1/2}$ and $\mathbf{I-D}^{-1/2}\mathbf{WD}^{1/2}$). 
\hide{
provided by 

\begin{eqnarray}
	\mathbf{W_0} & = & \mathbf{D}^{-1/2}\mathbf{WD}^{1/2}  \\
	\mathbf{L}_0 & = & \mathbf{D-W}_0 = \mathbf{I-D}^{-1/2}\mathbf{WD}^{1/2}
\end{eqnarray} 
}

\subsection{Model II: DP-Walker}
Our second method, Degree Penalty based Random Walk (DP-Walker), utilizes a skip-gram model, word2vec. 

We start with a brief description of the word2vec model and its applications on network embedding tasks. 
Given a text corpus, Mikolov et al. proposed word2vec to learn the representation of each word by facilitating the prediction of its context.
Inspired by it, some network embedding algorithms like DeepWalk~\cite{perozzi2014deepwalk} and node2vec~\cite{grover2016node2vec} define a vertex's ``context'' by co-occurred vertexes in random walk generated paths. 

Specifically, the Random Walk rooted at vertex $v_i$ is a stochastic process ${\mathcal{V}_{v_i}^k, k = 1,\ldots, m}$, where $\mathcal{V}_{v_i}^{k+1}$ is a vertex sampled from the neighbors of $\mathcal{V}_{v_i}^k$ and m is the path length.  
Traditional methods regard $P(\mathcal{V}_{v_i}^{k+1}|\mathcal{V}_{v_i}^{k})$ as a uniform distribution where each neighbor of $\mathcal{V}_{v_i}^{k}$ has the equal chance to be sampled. 

However, as our proposed Degree Penalty principle suggests, a neighbor $v_i$ of a vertex $v_j$ with high degree may not be similar to $v_j$. In other words, $v_j$ shall have less chance to be sampled as one of $v_i$'s context. 
Thus, we define the probability of the random walk jumping from $v_i$ to $v_j$ as 

\beq{
\label{eq:deepwalk}
	\rm{Pr}(v_j|v_i) \propto \frac{C'_{ij}}{(D_{ii} D_{jj})^\beta} 
} 

\noindent where $\mathbf{C}'$ can be found in Eq.~\ref{eq:matrix_c} and $\beta$ is the model parameter. 
According to Eq.~\ref{eq:deepwalk}, we find that $v_j$ will have a greater chance to be sampled when it has more common neighbors with $v_i$ and has a lower degree.
After obtaining random walk generated paths, we enable skip-gram to learn effective vertex representations for $G$ by predicting each vertex's context. This results in an optimization problem

\beq{
\label{eq:objective_dpw}
\argmin_{\mathbf{U}} -\log \rm{Pr}\left( \{ v_{i-w},\ldots, v_{i+w} \} \ v_i | \mathbf{u_i} \right)
}

\noindent where $2\times w$ is the path length we consider.
Specifically, for each random walk $\mathcal{V}_{v_i}$, we feed it to skip-gram algorithm~\cite{mikolov2013efficient} to update the vertex representations.
For implementation, we use Hierarchical Softmax~\cite{Morin05hierarchicalprobabilistic,Mnih2009hier} to estimate the concerned probability distribution. 

\section{Experiments}
\label{sec:exp}

\subsection{Experiment Setup}
\label{sec:exp_setup}

\vpara{Datasets.}
We use four datasets, whose statistics are summarized in Table~\ref{tb:data} for the evaluations. 

\begin{itemize}
\item Synthetic: We generate a synthetic dataset by the Preferential Attachment model~\cite{vazquez2003growing}, which describes the generation of scale-free networks.  
\item Facebook~\cite{leskovec2012learning}: This dataset is a subnetwork of Facebook\footnote{http://facebook.com}, where vertexes indicate users of Facebook, and edges denote friendship between users. 
\item Twitter~\cite{leskovec2012learning}: This dataset is a subnetwork of Twitter\footnote{http://twitter.com}, where vertexes indicate users of Twitter, and edges denote following relationships.
\item Coauthor~\cite{leskovec2007graph}: This network covers scientific collaborations between authors. Vertexes are authors.  Vertexes are authors. An undirected edge exists between two authors if they have coauthored at least one paper. 
\item Citation~\cite{tang2008arnetminer}: Similar to Coauthor, this is also an academic network, where vertexes are authors. Edges indicate citations instead of coauthor relationship. 
\item Mobile: This is a mobile network provided by PPDai\footnote{The largest unsecured micro-credit loan platform in China.}. Vertexes are PPDai registered users. An edge between two users indicates that one of the users has called the other. Overall, it consists of over one million calling logs.
\end{itemize}

\begin{table}[t]
\centering
\resizebox{0.9\linewidth}{!}{
\begin{tabular}{ccccccc}
\toprule
	& Synthetic &  Facebook &  Twitter & Coauthor &  Citation  &  Mobile\\
	\midrule
	$|V|$ &   \small 10000 &  \small 4039 &  \small 81306 &  \small 5242 &  \small 48521 &  \small 198959 \\
  \midrule
  $|E|$ & \small 399580 &  \small 88234 &  \small 1768149 &  \small 28980 &  \small 357235 &  \small 1151003 \\
\bottomrule
\end{tabular}
\normalsize
}
\caption{Statistics of datasets. 
\small $|V|$ indicates the number of vertexes and $|E|$ denotes the number of edges. \normalsize
}
\label{tb:data}
\end{table}




\vpara{Baseline methods.} We compare the following four network embedding algorithms in our experiments: 
\begin{itemize}
\item Laplacian Eigenmap (LE)~\cite{belkin2003laplacian}:  This represents spectral-based network embedding algorithms. It aims to learn the low-dimensional representation to expand the manifold where the data lie.
\item DeepWalk~\cite{perozzi2014deepwalk}: This represents skip-gram model based network embedding algorithms. It first generates random walks on the network, and defines the context of a vertex by its co-occurred vertexes in paths. Then, it learns vertex representations by predicting each vertex's context. Specifically, we perform $10$ random walks starting from each vertex, and each random walk will have a length of $40$. 
\item DP-Spectral: This is a spectral technique based implementation of our degree penalty principle. 
\item DP-Walker: This is another implementation of our approach. It is based on a skip-gram model. 
\end{itemize}
Unless otherwise specified, the embedding dimension for our experiments is $200$. 

\vpara{Tasks.}
We first utilize different algorithms to learn vertex representations for a certain dataset, then apply the learned embeddings in three different tasks: 
\begin{itemize}
\item Network reconstruction: this task aims to validate if an algorithm is able to preserve the scale-free property of networks. We evaluate the performance of different algorithms by the correlation coefficients between the reconstructed degrees and the degrees in the given network. 
\item Link prediction: given two vertexes $v_i$ and $v_j$, we feed their embedding vectors, $\mathbf{u_i}$ and $\mathbf{u_j}$, to a linear regression classifier and determine whether there exists an edge between $v_i$ and $v_j$. Specifically, we use $\mathbf{u}_i-\mathbf{u}_j$ as the feature, and randomly select about $1\%$ pairs of vertexes for training and evaluation. 
\item Vertex classification: on this task, we consider the vertex labels. For instance, in Citation, each vertex has a label to indicate the researcher's research area. Specifically, given a vertex $v_i$, we define its feature vector as $\mathbf{u_i}$, and train a linear regression classifier to determine its label.  
\end{itemize}
\subsection{Network Reconstruction}
\label{sec:exp:scalefree} 

\begin{table}[t]
\centering
\resizebox{\linewidth}{!}{
\begin{tabular}{|l|l|c|c|c|}
\hline
Dataset & Method ($\varepsilon$) & Pearson & Spearman & Kendall \\
\hline
 \multirow{4}*{Synthetic}  & LE (0.55) & 0.14 & 0.054 & 0.039 \\ \cline{2-5}
 & DeepWalk (0.91) & 0.47 & -0.22 & -0.18  \\ \cline{2-5}
  & DP-Spectral (0.52) & 0.92 & \textbf{0.79} & \textbf{0.63} \\ \cline{2-5}
  & DP-Walker (0.95) & \textbf{0.94} & 0.63 & 0.52 \\
 \hline
 
\multirow{4}*{Facebook} & LE (0.52) & 0.48 & 0.18 & 0.12 \\ \cline{2-5}
 & DeepWalk (0.81) & 0.73 & 0.65 & 0.49 \\ \cline{2-5} 
 
 & DP-Spectral (0.52) & \textbf{0.87} & \textbf{0.67} & 0.51 \\ \cline{2-5}
 
 & DP-Walker (0.84) & 0.75 & 0.73 & \textbf{0.57} \\ 
 \hline
 
 \multirow{4}*{Twitter}  & LE (0.81) & 0.17 & 0.19 & 0.17 \\ \cline{2-5}
 & DeepWalk (0.51) & 0.08 & 0.21 & 0.26 \\ \cline{2-5}
  & DP-Spectral (0.93) & \textbf{0.50} & \textbf{0.34} & \textbf{0.27} \\ \cline{2-5}
 & DP-Walker (0.087) & 0.40 & 0.33 & \textbf{0.27} \\ 
 \hline

 \multirow{4}*{Coauthor}  & LE (0.50) & 0.32 & 0.04 & 0.03 \\ \cline{2-5}
	& DeepWalk (0.92) & 0.66 & 0.31 & 0.25 \\ \cline{2-5}
  & DP-Spectral (0.51) & 0.64 & \textbf{0.69} & \textbf{0.55} \\ \cline{2-5}
  & DP-Walk (0.93) & \textbf{0.75} & 0.44 & 0.35 \\ 
 \hline
 \multirow{4}*{Citation}  & LE (0.99) & 0.11 & -0.27 & -0.20 \\ \cline{2-5}
 & DeepWalk (0.97) & 0.51 & 0.28 & 0.20 \\ \cline{2-5}
  & DP-Spectral (0.50) & 0.45 & \textbf{0.72} & \textbf{0.54} \\ \cline{2-5}
 & DP-Walker (0.98) & \textbf{0.62} & 0.65 & 0.50 \\ 
 \hline
  \multirow{4}*{Mobile}  & LE (0.51) & 0.10 & 0.05 & 0.04 \\ \cline{2-5}
 & DeepWalk (0.71) & 0.77 & 0.22 & 0.20 \\ \cline{2-5}
  & DP-Spectral (0.50) & 0.40 & \textbf{0.68} & \textbf{0.60} \\ \cline{2-5}
 & DP-Walker (0.78) & \textbf{0.93} & 0.22 & 0.20 \\ 
 \hline
 
\end{tabular}
}
\caption{Performance of different methods on scale-free property reconstruction. 
For each method, $\varepsilon$ (indicated after the corresponding method) is chosen so that Pearson is maximized. 
\normalsize
}
\label{tb:results:correlations}
\end{table}

\begin{figure*}[htbp]
\centering
\subfigure[Synthetic.]{
\label{fig:dim:PA}
\includegraphics{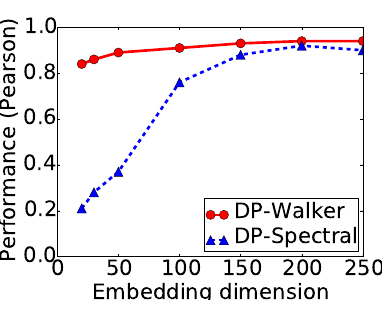}
}
\subfigure[Facebook.]{
\label{fig:dim:FB}
\includegraphics{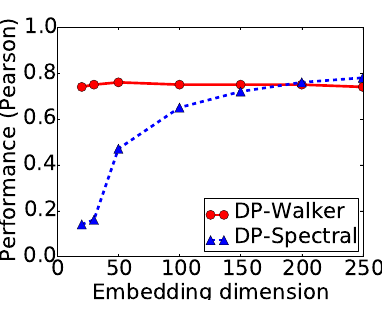}
}
\subfigure[Synthetic.]{
\label{fig:beta:PA}
\includegraphics{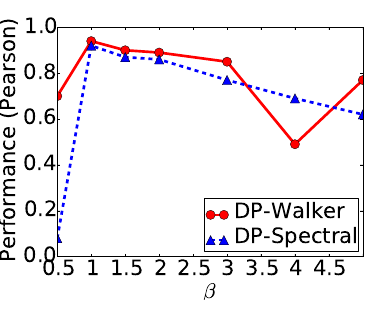}
}
\subfigure[Facebook.]{
\label{fig:beta:FB}
\includegraphics{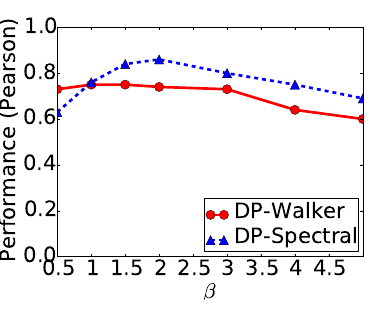}
}
\caption{Model parameter analysis. (a) and (b) demonstrate the sensitivity of the embedding dimension $k$ in Synthetic and Facebook dataset respectively. (c) and (d) present the sensitivity of the degree penalty parameter $\beta$. We omit the results on other datasets due to space limitation. }
\label{fig:beta}
\end{figure*}
\vpara{Comparison results.} In this task, after obtaining vertex representations, we use the $\varepsilon$-NN algorithm introduced in Section~\ref{sec:approach} to reconstruct the network. 
We then evaluate the performance of different methods on reconstructing power-law distributed vertex degrees by considering three different correlation coefficients of the reconstructed degrees and original degrees: Pearson's $r$~\cite{statistics}, Spearman's $\rho$~\cite{spearman}, and Kendall's $\tau$~\cite{kendalltau}. All of these statistics are used to evaluate some relationship between paired samples. Pearson's $r$ is used to detect linear relationships, while Spearman's $\rho$ and Kendall's $\tau$ are capable of finding monotonic relationships. 

To select $\varepsilon$, for each algorithm, we roll over all values of $\varepsilon$ ranging from $0.01$ to $1$ with step $0.01$, and choose the value which maximizes the Pearson's correlation coefficient between the original and recovered degrees. Our selection of Pearson's $r$ as evaluation metric is because of the scale-free property of degree distribution. 

From Table~\ref{tb:results:correlations}, we see that our algorithms outperform the baselines significantly. 
Especially, Pearson's $r$ of DP-Walker achieves at least $44.5\%$ improvement, and Spearman's $\rho$ of DP-Spectral achieves at least $84.3\%$ improvement. 
The good fitness of the vertex degree reconstructed by our algorithms suggests that we can better preserve the scale-free property of the network. 
Besides, we see that the best $\varepsilon$ to maximize Pearson's $r$ for DP-Spectral is more stable (i.e., around $0.51$) than other methods. 
Thus one can tune DP-Spectral's parameters more easily in practice. 

\vpara{Preserving Scale-free Property.} 
After reconstructing the network, for each embedding algorithm, we validate its effectiveness to preserve the scale-free property, by fitting the reconstructed degree distribution (e.g., Figure~\ref{fig:deg_citation_recovered_unmod} and (c)) to a theoretical power-law distribution~\cite{alstott2014powerlaw}. 
We then calculate the Kolmogorov-Smirnov distance between these two distributions and find that the proposed methods can always obtain better results compared to baselines (0.115 vs 0.225 averagely). 

\vpara{Parameter analysis.} We further study the sensitivity of model parameters: embedding dimension and degree penalty weight $\beta$. 
We only present the results on Synthetic and Facebook and omit others due to space limitation. 
Figure~\ref{fig:dim:PA} and \ref{fig:dim:FB} shows the Pearson correlation coefficients resulted by our algorithm with different embedding dimensions. 
When the embedding dimension grows, the performance increases significantly. The figures suggest that when embedding a network into Euclidean space, there is a dimension which is sufficient for preserving scale-free property, and further increase of the embedding dimension has limited effect.
Correlation does not change drastically for DP-Walker as the dimension increases, as the figure suggests. 
The figures also show that DP-Walker is able to obtain fairly high Pearson correlation even in a lower dimension and it requires an embedding dimension lower than DP-Spectral does for a satisfactory performance. 

We also study how $\beta$ influences the performance. Figure~\ref{fig:beta:PA} and \ref{fig:beta:FB} shows that DP-Spectral is more sensitive to the choice of $\beta$. This is largely due to the fact that in DP-Spectral, the influence of $\beta$ is manifested in the objective function (Eq.~\ref{eq:objective_dps}), which imposes a stronger restriction than its counterpart in DP-Walker, which is embodied in the sampling process. 
Figure~\ref{fig:beta:PA} and \ref{fig:beta:FB} shows that the optimal choice of $\beta$ varies from graph to graph. It makes sense, since $\beta$ can be viewed as a punishment on the degrees, and the influence of $\beta$ is therefore supposedly related to the topology of the original network. 


\begin{table}[t]
\centering
\resizebox{\linewidth}{!}{
\begin{tabular}{|c|c|c|c|c|}
\hline
\multirow{1}*{Dataset} & \multirow{1}*{Method} & \multirow{1}*{Precision} & \multirow{1}*{\ \ Recall \ \ } & \multirow{1}*{\ \ \ \ F1 \ \ \ \  } \\ 
\hline 
\multirow{4}*{Synthetic} 
& LE & 0.52	& 0.53	& 0.53 \\ \cline{2-5}
& Deepwalk & 0.51 &	0.51 & 0.51  \\ \cline{2-5}
& DP-Spectral & \textbf{0.64} &	 \textbf{0.68} & \textbf{0.66} \\ \cline{2-5}
& DP-Walker & 0.61 & 0.63	& 0.62 \\ \hline

\multirow{4}*{\ Facebook \ } 
 & LE & 0.75 & 0.92 & 0.83 \\ \cline{2-5}
 & Deepwalk & \textbf{0.84} & 0.97 &	 \textbf{0.90} \\ \cline{2-5}
 & DP-Spectral & 0.76 &	\textbf{0.98}	 & 0.85 \\ \cline{2-5}
 & DP-Walker & 0.82& 0.95 & 0.89 \\ \hline

\multirow{4}*{Twitter} 
& LE & 0.58  & 0.35 & 0.43 \\ \cline{2-5}
& Deepwalk & \textbf{0.65} & 	0.77 &	0.70 \\ \cline{2-5}
& DP-Spectral & 0.59 & \textbf{0.98} & \textbf{0.73} \\ \cline{2-5}
& DP-Walker  & 0.54 & 0.58 & 0.56  \\ \hline

\multirow{4}*{Coauthor} 
& LE & 0.61 & 0.83 & 0.70 \\ \cline{2-5}
& Deepwalk & 0.55 &	0.58	 & 0.56 \\ \cline{2-5}
&\ \  DP-Spectral \ \  & \textbf{0.62} & \textbf{0.89} &  \textbf{0.73}  \\ \cline{2-5}
& DP-Walker & 0.56 	& 0.66	& 0.61  \\ \hline

\multirow{4}*{Citation} 
& LE & 0.54 & 0.56 & 0.55 \\ \cline{2-5}
& Deepwalk & 0.54 &	0.56 & 0.55 \\ \cline{2-5}
& DP-Spectral & 0.52 &\textbf{0.99} & \textbf{0.68} \\ \cline{2-5}
& DP-Walker  & \textbf{0.55} &	0.57 & 0.56   \\ \hline

\multirow{4}*{Mobile} 
& LE & \textbf{0.75}  & 0.36 & 0.48 \\ \cline{2-5}
& Deepwalk & 0.55 &	0.60 &	0.57 \\ \cline{2-5}
& DP-Spectral & 0.63 & \textbf{0.89} & \textbf{0.74} \\ \cline{2-5}
& DP-Walker  & 0.54 & 0.58 & 0.56  \\ \hline

\end{tabular}
}
\caption{Experimental results of link prediction. 
}
\label{tb:exp:link}
\end{table}

\subsection{Link Prediction}
\label{sec:exp:link}
In this task, we consider the following evaluation metrics: Precision, Recall, and F1-score. 
Table~\ref{tb:exp:link} demonstrates the performance of different methods on the link prediction task. 
For our methods, we use	 the model parameter $\beta$ as optimized in Table \ref{tb:results:correlations}. 
We see that, in most cases, DP-Spectral obtains the best performance, which suggests that with the help of the proposed principle, we can not only preserve the scale-free property of networks, but also improve the effectiveness of the embedding vectors. 
\subsection{Vertex Classification}
\label{sec:exp:vertex} 
Table~\ref{tb:exp:multi} lists the accuracy of vertex classification task on Citation.
Our task is to determine an author's research area, which is a multi-classification problem.
We define features as vertex representation obtained by the four different embedding algorithms. 
Generally, from the table, we see that DP-Walker and DP-Spectral beat respectively Deepwalk and Laplacian Eigenmap. 
In particular, DP-Spectral achieves the best result for 5 out of 7 labels. 
Besides, we can also observe its stability of the performance. 
DP-Spectral's result on all labels is more stable than other methods. 
In comparison, LE achieves a satisfactory result for two labels, but for others the result can be poor. 
Specifically, the standard deviation of DP-Spectral is 0.04, while the value is 0.26 for LE and 0.1 for DeepWalk.


\section{Related Work}
\label{sec:related}

\begin{table}[t]
	\resizebox{\linewidth}{!}{
	\begin{tabular}{|c|c|c|c|c|c|c|c|} \hline
		Method & Acrh & CN & CS & DM & THM & GRA & UNK
		\\ \hline
		LE & 0.36 & \textbf{0.75} & 0.14 & 0.37 & 0.46 & 0.13 & \textbf{0.86}
		\\ \hline
		Deepwalk & 0.54 & 0.54 & 0.52 & 0.56 & 0.56 & 0.56 & 0.85
		\\ \hline
		DP-Walker & 0.56 & 0.57 & 0.54 & 0.58 & 0.58 & 0.55 & 0.85
		\\ \hline		
		DP-Spectral & \textbf{0.71} & 0.74 & \textbf{0.78} & \textbf{0.76} & \textbf{0.74} & \textbf{0.75} & 0.85
		\\ \hline
		\end{tabular}
	}
	\caption{Accuracy of multi-classification task. 
	The labels stand Architecture, Computer Network, Computer Science, Data Mining, Theory, Graphics, and Unknown, respectively.	}
	\label{tb:exp:multi}
\end{table}
\paragraph{Network embedding.} 
Network embedding aims to learn representations for vertexes in a given network. 
Some researchers regard network embedding as part of dimensionality reduction techniques.
For example, Laplacian Eigenmaps (LE)~\cite{belkin2003laplacian} aims to learn the vertex representation to expand the manifold  where the data lie. 
As a variant of LE, Locality Preserving Projections (LPP)~\cite{he2005face} learns a linear projection from feature space to embedding space. 
Besides, there are other linear~\cite{jolliffe2002principal} and non-linear~\cite{tenenbaum2000global} network embedding algorithms for dimensionality reduction.
Recent network embedding works take advancements in natural language processing, most notably models known as word2vec~\cite{mikolov2013efficient}, which learns the distributed representations of words. 
Building on word2vec, Perozzi et al. define a vertex's ``context'' by their co-occurrence in a random walk path~\cite{perozzi2014deepwalk}. 
More recently, Grover et al. propose a mixture of width-first and breadth-first search based procedure to generate paths of vertexes~ \cite{grover2016node2vec}. 
Dong et al. further develop the model to handle heterogeneous networks~\cite{dong2017metapath}. 
LINE~\cite{tang2015line}
decomposes a vertex's context into first-order (neighbors) and second-order (two-degree neighbors) proximity. 
Wang et al. preserve community information in their vertex representations~\cite{wang2017community}. 
However, all of above methods focus on preserving microscopic network structure and ignore macroscopic scale-free property of networks. 

\vpara{Scale-free Networks.}
The scale-free property has been discovered to be ubiquitous over a variety of network systems~\cite{mood1950introduction,newman2005power,clauset2009power}, such as 
the Internet Autonomous System graph~\cite{faloutsos1999power}, the Internet router graph~\cite{govindan2000heuristics}, the degree 
distributions of subsets of the world wide web~\cite{barabasi1999emergence}. 
\citeauthor{newman2005power} provides a comprehensive list of such work~\cite{newman2005power}.
However, investigating the scale-free property in a low-dimensional vector space and establishing its cooperation 
with network embedding 
have not been fully considered.

\section{Conclusion}
\label{sec:conclusion}
In this paper, we study the problem of learning the scale-free property preserving network embeddings. 
We first analyze the feasibility of reconstructing a scale-free network based on learned vertex representations in the Euclidean space by converting our problem to the Sphere Packing problem. 
We then propose the \textit{degree penalty} principle as our solution and introduce two implementations by leveraging spectral techniques and a skip-gram model respectively. 
The proposed principle can also be implemented using other methods, which is left as our future work. 
We at last conduct extensive experiments on both synthetic data and five real-world datasets to verify the effectiveness of our approach.

\paragraph{Acknowledgements.}
The work is supported by the Fundamental Research Funds for the Central Universities, 973 Program (2015CB352302), NSFC (U1611461, 61625107, 61402403), and key program of Zhejiang Province (2015C01027).
\normalsize

\balance
\bibliographystyle{aaai}
\bibliography{mybibtex}

\begin{thebibliography}{}

\bibitem[\protect\citeauthoryear{Adamic and Huberman}{1999}]{adamic1999nature}
Adamic, L., and Huberman, B.~A.
\newblock 1999.
\newblock The nature of markets in the world wide web.
\newblock {\em Q. J. Econ.}

\bibitem[\protect\citeauthoryear{Alanis-Lobato, Mier, and
  Andrade-Navarro}{2016}]{alanis2016labne}
Alanis-Lobato, G.; Mier, P.; and Andrade-Navarro, M.~A.
\newblock 2016.
\newblock Efficient embedding of complex networks to hyperbolic space via their
  laplacian.
\newblock In {\em Scientific reports}.

\bibitem[\protect\citeauthoryear{{Alstott}, {Bullmore}, and
  {Plenz}}{2014}]{alstott2014powerlaw}
{Alstott}, J.; {Bullmore}, E.; and {Plenz}, D.
\newblock 2014.
\newblock {powerlaw: A Python Package for Analysis of Heavy-Tailed
  Distributions}.
\newblock {\em PLoS ONE} 9:e85777.

\bibitem[\protect\citeauthoryear{Barab{\'a}si and
  Albert}{1999}]{barabasi1999emergence}
Barab{\'a}si, A.-L., and Albert, R.
\newblock 1999.
\newblock Emergence of scaling in random networks.
\newblock {\em science} 286(5439):509--512.

\bibitem[\protect\citeauthoryear{Belkin and Niyogi}{2003}]{belkin2003laplacian}
Belkin, M., and Niyogi, P.
\newblock 2003.
\newblock Laplacian eigenmaps for dimensionality reduction and data
  representation.
\newblock {\em Neural Comput.} 15(6):1373--1396.

\bibitem[\protect\citeauthoryear{Clauset, Shalizi, and
  Newman}{2009a}]{aaron2009}
Clauset, A.; Shalizi, C.~R.; and Newman, M. E.~J.
\newblock 2009a.
\newblock Power-law distributions in empirical data.
\newblock {\em SIAM Review} 51(4):661--703.

\bibitem[\protect\citeauthoryear{Clauset, Shalizi, and
  Newman}{2009b}]{clauset2009power}
Clauset, A.; Shalizi, C.~R.; and Newman, M.~E.
\newblock 2009b.
\newblock Power-law distributions in empirical data.
\newblock {\em SIAM review} 51(4):661--703.

\bibitem[\protect\citeauthoryear{Cohn and Elkies}{2003}]{cohn2003new}
Cohn, H., and Elkies, N.
\newblock 2003.
\newblock New upper bounds on sphere packings i.
\newblock {\em Annals of Mathematics}  689--714.

\bibitem[\protect\citeauthoryear{Cohn and Zhao}{2014}]{cohn2014sphere}
Cohn, H., and Zhao, Y.
\newblock 2014.
\newblock Sphere packing bounds via spherical codes.
\newblock {\em Duke Mathematical Journal} 163(10):1965--2002.

\bibitem[\protect\citeauthoryear{Dong, Chawla, and
  Swami}{2017}]{dong2017metapath}
Dong, Y.; Chawla, N.; and Swami, A.
\newblock 2017.
\newblock metapath2vec: Scalable representation learning for heterogeneous
  networks.
\newblock In {\em KDD'17},  135--144.

\bibitem[\protect\citeauthoryear{Faloutsos, Faloutsos, and
  Faloutsos}{1999}]{faloutsos1999power}
Faloutsos, M.; Faloutsos, P.; and Faloutsos, C.
\newblock 1999.
\newblock On power-law relationships of the internet topology.
\newblock In {\em COMPUT COMMUN REV}, volume~29,  251--262.

\bibitem[\protect\citeauthoryear{Govindan and
  Tangmunarunkit}{2000}]{govindan2000heuristics}
Govindan, R., and Tangmunarunkit, H.
\newblock 2000.
\newblock Heuristics for internet map discovery.
\newblock In {\em INFOCOM'00}, volume~3,  1371--1380.

\bibitem[\protect\citeauthoryear{Grover and
  Leskovec}{2016}]{grover2016node2vec}
Grover, A., and Leskovec, J.
\newblock 2016.
\newblock node2vec: Scalable feature learning for networks.
\newblock In {\em KDD'16},  855--864.

\bibitem[\protect\citeauthoryear{He \bgroup et al\mbox.\egroup
  }{2005}]{he2005face}
He, X.; Yan, S.; Hu, Y.; Niyogi, P.; and Zhang, H.
\newblock 2005.
\newblock Face recognition using laplacianfaces.
\newblock {\em IEEE Transactions on Pattern Analysis and Machine Intelligence}
  27(3):328--340.

\bibitem[\protect\citeauthoryear{Jolliffe}{2002}]{jolliffe2002principal}
Jolliffe, I.
\newblock 2002.
\newblock {\em Principal component analysis}.
\newblock Wiley Online Library.

\bibitem[\protect\citeauthoryear{Kabatiansky and
  Levenshtein}{1978}]{kabatiansky1978bounds}
Kabatiansky, G.~A., and Levenshtein, V.~I.
\newblock 1978.
\newblock On bounds for packings on a sphere and in space.
\newblock {\em Problemy Peredachi Informatsii} 14(1):3--25.

\bibitem[\protect\citeauthoryear{Kendall}{1938}]{kendalltau}
Kendall, M.~G.
\newblock 1938.
\newblock A new measure of rank correlation.
\newblock {\em Biometrika} 30(1/2):81--93.

\bibitem[\protect\citeauthoryear{Leskovec and
  Mcauley}{2012}]{leskovec2012learning}
Leskovec, J., and Mcauley, J.
\newblock 2012.
\newblock Learning to discover social circles in ego networks.
\newblock In {\em NIPS'12},  539--547.

\bibitem[\protect\citeauthoryear{Leskovec, Kleinberg, and
  Faloutsos}{2007}]{leskovec2007graph}
Leskovec, J.; Kleinberg, J.~M.; and Faloutsos, C.
\newblock 2007.
\newblock Graph evolution: Densification and shrinking diameters.
\newblock {\em ACM Transactions on Knowledge Discovery From Data} 1(1):2.

\bibitem[\protect\citeauthoryear{Mikolov \bgroup et al\mbox.\egroup
  }{2013}]{mikolov2013efficient}
Mikolov, T.; Chen, K.; Corrado, G.; and Dean, J.
\newblock 2013.
\newblock Efficient estimation of word representations in vector space.
\newblock In {\em NIPS'13},  3111--3119.

\bibitem[\protect\citeauthoryear{Mnih and Hinton}{2009}]{Mnih2009hier}
Mnih, A., and Hinton, G.~E.
\newblock 2009.
\newblock A scalable hierarchical distributed language model.
\newblock In Koller, D.; Schuurmans, D.; Bengio, Y.; and Bottou, L., eds., {\em
  Advances in Neural Information Processing Systems 21}. Curran Associates,
  Inc.
\newblock  1081--1088.

\bibitem[\protect\citeauthoryear{Mood}{1950}]{mood1950introduction}
Mood, A.~M.
\newblock 1950.
\newblock Introduction to the theory of statistics.

\bibitem[\protect\citeauthoryear{Morin and
  Bengio}{2005}]{Morin05hierarchicalprobabilistic}
Morin, F., and Bengio, Y.
\newblock 2005.
\newblock Hierarchical probabilistic neural network language model.
\newblock In {\em AISTATS’05},  246--252.

\bibitem[\protect\citeauthoryear{Newman}{2005}]{newman2005power}
Newman, M.~E.
\newblock 2005.
\newblock Power laws, pareto distributions and zipf's law.
\newblock {\em Contemporary physics} 46(5):323--351.

\bibitem[\protect\citeauthoryear{Perozzi, Al-Rfou, and
  Skiena}{2014}]{perozzi2014deepwalk}
Perozzi, B.; Al-Rfou, R.; and Skiena, S.
\newblock 2014.
\newblock Deepwalk: Online learning of social representations.
\newblock In {\em KDD'14},  701--710.

\bibitem[\protect\citeauthoryear{Shao}{2007}]{statistics}
Shao, J.
\newblock 2007.
\newblock {\em Mathematical Statistics}.
\newblock Springer, 2 edition.

\bibitem[\protect\citeauthoryear{Shaw and Jebara}{2009}]{shaw2009spe}
Shaw, B., and Jebara, T.
\newblock 2009.
\newblock Structure preserving embedding.
\newblock In {\em Proceedings of the 26th Annual International Conference on
  Machine Learning}, ICML '09,  937--944.
\newblock New York, NY, USA: ACM.

\bibitem[\protect\citeauthoryear{Spearman}{1904}]{spearman}
Spearman, C.
\newblock 1904.
\newblock The proof and measurement of association between two things.
\newblock {\em The American Journal of Psychology} 15(1):72--101.

\bibitem[\protect\citeauthoryear{Tang \bgroup et al\mbox.\egroup
  }{2008}]{tang2008arnetminer}
Tang, J.; Zhang, J.; Yao, L.; Li, J.; Zhang, L.; and Su, Z.
\newblock 2008.
\newblock Arnetminer: extraction and mining of academic social networks.
\newblock In {\em KDD'08},  990--998.

\bibitem[\protect\citeauthoryear{Tang \bgroup et al\mbox.\egroup
  }{2015}]{tang2015line}
Tang, J.; Qu, M.; Wang, M.; Zhang, M.; Yan, J.; and Mei, Q.
\newblock 2015.
\newblock Line: Large-scale information network embedding.
\newblock In {\em WWW'15},  1067--1077.

\bibitem[\protect\citeauthoryear{Tenenbaum, De~Silva, and
  Langford}{2000}]{tenenbaum2000global}
Tenenbaum, J.~B.; De~Silva, V.; and Langford, J.~C.
\newblock 2000.
\newblock A global geometric framework for nonlinear dimensionality reduction.
\newblock {\em science} 290(5500):2319--2323.

\bibitem[\protect\citeauthoryear{Vance}{2011}]{vance2011sphere}
Vance, S.
\newblock 2011.
\newblock Improved sphere packing lower bounds from hurwitz lattices.
\newblock {\em Advances in Mathematics} 227(5):2144--2156.

\bibitem[\protect\citeauthoryear{Vazquez}{2003}]{vazquez2003growing}
Vazquez, A.
\newblock 2003.
\newblock Growing network with local rules: Preferential attachment, clustering
  hierarchy, and degree correlations.
\newblock {\em Physical Review E} 67(5):056104.

\bibitem[\protect\citeauthoryear{Venkatesh}{2012}]{venkatesh2012sphere}
Venkatesh, A.
\newblock 2012.
\newblock A note on sphere packings in high dimension.
\newblock {\em International mathematics research notices} 2013(7):1628--1642.

\bibitem[\protect\citeauthoryear{Wang \bgroup et al\mbox.\egroup
  }{2017}]{wang2017community}
Wang, X.; Cui, P.; Wang, J.; Pei, J.; Zhu, W.; and Yang, S.
\newblock 2017.
\newblock Community preserving network embedding.
\newblock In {\em AAAI'17}.

\end{thebibliography}

\end{document}